\documentclass[10pt,conference]{IEEEtran}
\IEEEoverridecommandlockouts
\IEEEpubid{%
  \makebox[\columnwidth][l]{%
    \begin{minipage}[t]{\columnwidth}\vspace{2pt}\footnotesize
      © 2025 IEEE. Personal use of this material is permitted. Permission from IEEE must be obtained for all other uses, in any current or future media, including reprinting/republishing this material for advertising or promotional purposes, creating new collective works, for resale or redistribution to servers or lists, or reuse of any copyrighted component of this work in other works.
    \end{minipage}%
  }%
  \hspace{\columnsep}%
  \makebox[\columnwidth]{}%
}

\usepackage{pifont}              
\newcommand{\cmark}{\ding{51}}   
\newcommand{\xmark}{\ding{55}}   
\usepackage{amsthm}

\usepackage[utf8]{inputenc}
\usepackage{enumerate}
\usepackage{graphicx}      
\usepackage{tabularx}   %
\usepackage{cuted} 
\usepackage{enumitem}
\usepackage{booktabs}      
\usepackage{amsmath,amsfonts,amssymb}  
\usepackage{algorithm}
\usepackage{quantikz}
\newtheorem{theorem}{Theorem}
\newtheorem{lemma}{Lemma}

\usepackage{algpseudocode}
\usepackage{cite}          
\usepackage{url}           
\usepackage{array}         
\usepackage{hyperref}      
\usepackage{caption}
\usepackage{subcaption}
\usepackage{tikz}
\usetikzlibrary{shapes.geometric, arrows, positioning}
\usepackage{float}         
\usepackage{amsmath,amssymb,amsthm}

\theoremstyle{remark}

\usepackage{enumitem}
\setlist[itemize,enumerate]{nosep,         
                            itemsep=0.5ex, 
                            leftmargin=*}  

\makeatletter
\g@addto@macro\normalsize{%
    \abovedisplayskip 1pt plus 1pt minus 1pt%
    \abovedisplayshortskip 1pt plus 1pt minus 1pt%
    \belowdisplayskip 1pt plus 1pt minus 1pt%
    \belowdisplayshortskip 1pt plus 1pt minus 1pt%
}
\makeatother

\def\BibTeX{{\rm B\kern-.05em{\sc i\kern-.025em b}\kern-.08em
    T\kern-.1667em\lower.7ex\hbox{E}\kern-.125emX}}

\title{Scalable Hardware Maturity Probe for Quantum Accelerators via Harmonic Analysis of QAOA}

\vspace{-0.5cm}
\author{
\IEEEauthorblockN{
Chinonso Onah\IEEEauthorrefmark{1}\IEEEauthorrefmark{2} and
Kristel Michielsen\IEEEauthorrefmark{2}\IEEEauthorrefmark{3}
\thanks{Corr. author: \texttt{chinonso.calistus.onah@volkswagen.de.}}
}

\IEEEauthorblockA{\IEEEauthorrefmark{1}Volkswagen Group, Wolfsburg, Germany}
\IEEEauthorblockA{\IEEEauthorrefmark{2}Department of Physics, RWTH Aachen University, Aachen, Germany}
\IEEEauthorblockA{\IEEEauthorrefmark{3}Jülich Supercomputing Centre, Institute for Advanced Simulation, Forschungszentrum Jülich, Jülich, Germany}

}
\IEEEaftertitletext{\vspace{-2.0\baselineskip}}

\begin{document}

\maketitle

\begin{abstract}
As quantum processors begin operating as tightly coupled accelerators inside high-performance computing (HPC) facilities, dependable and reproducible behaviour becomes a gating requirement for scientific and industrial workloads.  We present a \emph{hardware-maturity probe} that quantifies a device’s reliability by testing whether it can repeatedly reproduce the \emph{provably global} optima of single-layer Quantum Approximate Optimization Algorithm (QAOA) circuits. Using harmonic-analysis techniques, we derive closed-form upper bounds on the number of stationary points in the \(p{=}1\) QAOA cost landscape for broad classes of combinatorial-optimization problems. The bounds yield an exhaustive yet low-overhead grid-sampling scheme whose outcomes are analytically verifiable.  The probe integrates reliability-engineering notions—run-to-failure statistics, confidence-interval estimation, and reproducibility testing—into a single, application-centric benchmark. By linking analytic guarantees to empirical performance, our framework supplies a standardised dependability metric for hybrid quantum–HPC (QHPC) workflows.
\end{abstract}


\begin{figure}[!htb]
\centering
\begin{tikzpicture}[node distance=6mm, every node/.style={font=\small}]
  \tikzstyle{box}       = [rectangle, rounded corners, minimum width=3.0cm,
                            draw=black, thick, fill=gray!10, align=center]
  \tikzstyle{startnode} = [box, fill=green!20]   
  \tikzstyle{endnode}   = [box, fill=red!20]     
  \tikzstyle{arrow}     = [->, thick, >=latex]

  \node[startnode] (start) {Start};

  \node[box, below=of start] (initN)
        {\textbf{Step 1: Initialise}\\[-2pt]
         Set problem size $N=N_{\min}$};

  \node[box, below=of initN] (angles)
        {\textbf{Step 2: Classical Solve}\\[-2pt]
         Compute $(\beta^\star,\gamma^\star)$ and\\
         ideal energy $E_{\text{ideal}}$};

  \node[box, below=of angles] (run)
        {\textbf{Step 3: Execute QAOA}\\[-2pt]
         Prepare circuit, sample $M$ shots\\
         (repeat $R$ times)};

  \node[box, below=of run] (compare)
        {\textbf{Step 4: Evaluate Metrics}\\[-2pt]
         • \textit{Compound error},
         $C = N_{2Q}\,\overline{\epsilon}$\\
         • \textit{Energy gap},
         $\Delta E = \lvert E_{\text{ideal}}-\!E_{\text{meas}}\rvert$};

  \node[box, below=of compare] (decide)
        {\textbf{Step 5: Decision}\\[-2pt]
         Stop if $C>1$ \textbf{or}\\
         $\Delta E>\tau_{E}$};

  \node[box, right=22mm of decide] (increase)
        {\textbf{Step 6: Enlarge $N$}\\[-2pt]
         $N \leftarrow N + \Delta N$};

  \node[endnode, below=of decide] (end) {Stop};

  \draw[arrow] (start)   -- (initN);
  \draw[arrow] (initN)   -- (angles);
  \draw[arrow] (angles)  -- (run);
  \draw[arrow] (run)     -- (compare);
  \draw[arrow] (compare) -- (decide);
  \draw[arrow] (decide)  -- node[above,sloped]{\small Yes}(end);
  \draw[arrow] (decide)  -- node[above]{\small else}(increase);
  \draw[arrow] (increase) |- (angles);
\end{tikzpicture}
\caption{\emph{Run-to-failure} workflow driven by compound-error
and energy-gap metrics.  The green box marks the starting point; the red
box indicates termination once $C>1$ or $\Delta E>\tau_{E}$.}
\label{fig:run_to_failure}
\end{figure}
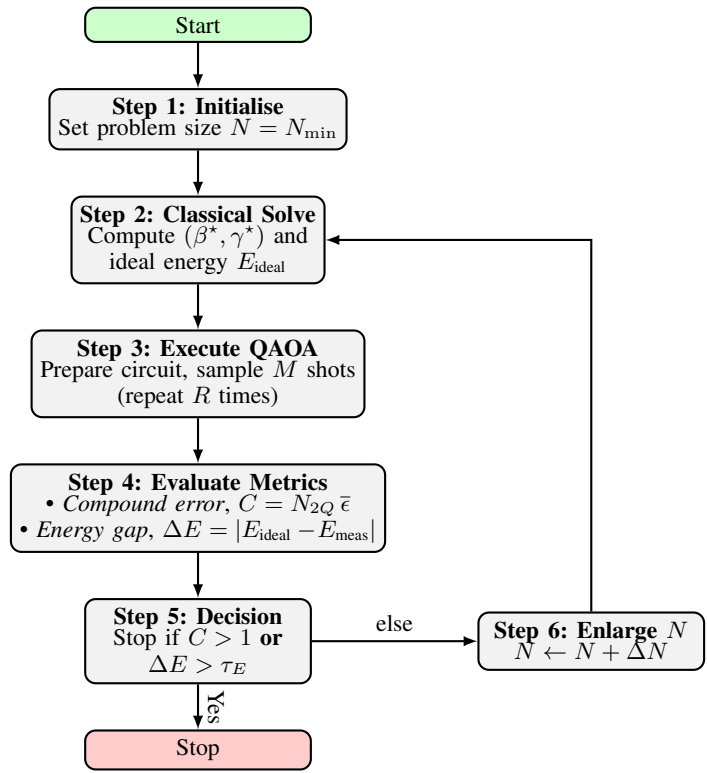

\section{Introduction}\label{sec:introduction}

Current Noisy Intermediate-Scale Quantum (NISQ) processors feature limited qubit counts and error rates well above fault-tolerance thresholds \cite{preskill2018quantum}.  Gate-level benchmarks like randomised benchmarking, Quantum Volume, EPLG, CLOPS, \emph{etc.} \cite{magesan2011scalable,EPLG21,Clops20}, characterise individual quantum operations, yet they offer little insight into whether a device can solve a \emph{specific} workload before noise overwhelms the answer.   Application oriented benchmark suites such as proposed in Refs.  \cite{Osaba_2023,MontanezBarrera2024b,Michielsen_2017,Willsch_2020,lubinski2023app, tomesh2022}, — use synthetic kernels or rely on classical co-simulation of problems. In addition, a scalable, problem-centred dependability metric is still missing.  This work proposes a single-layer QAOA probe that closes this gap. It directly embeds the \emph{user’s} problem graph, admits closed-form optima, and remains scalable even beyond 100 qubits, thereby filling the gap between device-level and domain-agnostic benchmarks.

As is well known, for a single-layer Quantum Approximate Optimisation Algorithm (QAOA) \cite{farhi2014quantum}, closed-form optima for the energy expectation vlaues exist for many combinatorial optimization problems. One can therefore bypass costly classical parameter optimisation by running the QAOA circuit with pre-computed angles.  Expectation values measured on the quantum hardwares or computed from bitstrings sampled from the quantum hardwares could be directly compared to analytic expectation values; furnishing a scalable, problem specific hardware probe.  This work defines two  task-level metrics to drive application specific hardware probe (Fig.~\ref{fig:run_to_failure}):

\begin{enumerate}[label=\arabic*)]
  \item \textbf{Compound error} $C = N_{2\mathrm{Q}}\,\overline{\epsilon}$, the product of the two-qubit-gate count and mean gate error.  Reliable operation demands $C < 1$.
  \item \textbf{Energy gap} $\Delta E = |E_{\text{ideal}}-E_{\text{meas}}|$, which flags when noise obscures the analytically optimal cost.
\end{enumerate}

If either $C>1$ or $\Delta E>\tau_{E}$ the benchmark halts; otherwise the instance size is increased.  The probe therefore links low-level noise to application-level fidelity while scaling to $>\!100$ qubits. This procedure constitutes a run-to-failure loop \cite{OConnor2012,Modarres2009} depicted in Figure~\ref{fig:run_to_failure} and remains scalable as long as the circuit initialization parameters and the circuit depth remain efficiently calculable. Our proposal therefore provides a scalable, problem-specific dependability metric for hybrid integration of quantum and high performance computing (QHPC) workflows.



Table~\ref{tab:benchmark_comparison} situates our proposal between gate-level and domain-agnostic suites, providing the first scalable, problem-specific dependability metric for hybrid QHPC workflows.

\begin{table}[ht]
\centering
\small
\resizebox{\columnwidth}{!}{%
\begin{tabular}{lcc}
\toprule
\textbf{Benchmark / Metric} & \textbf{Problem–specific?} & \textbf{Scales $>\!50$ qubits?}\\
\midrule
Randomised Benchmarking\,\cite{magesan2011scalable} & \textcolor{red}{\xmark} & \textcolor{green}{\cmark}\\
App. Oriented\,\cite{lubinski2023app,tomesh2022}                  & \textcolor{green}{\cmark} & \textcolor{red}{\xmark}\\
Monta\~nez-Barrera \emph{et al.}\,\cite{MontanezBarrera2024b} & \textcolor{green}{\cmark} & \textcolor{red}{\xmark}\\
Michielsen \emph{et al.}\,\cite{Michielsen_2017}    & \textcolor{green}{\cmark} & \textcolor{red}{\xmark}\\
Willsch \emph{et al.}\,\cite{Willsch_2020}          & \textcolor{green}{\cmark} & \textcolor{red}{\xmark}\\
\midrule
\textbf{This work (QAOA Probe)}                     & \textcolor{green}{\cmark} & \textcolor{green}{\cmark}\\
\bottomrule
\end{tabular}}
\caption{Gate-level vs application-level benchmarks. Green and red  mark positive and negative attributes respectively.}
\label{tab:benchmark_comparison}
\end{table}

\vspace{-0.3cm}

\section{Methods}

\vspace{-0.2cm}

\subsection{Quantum Approximate Optimization Algorithm (QAOA)}
QAOA~\cite{farhi2014quantum} is a hybrid quantum-classical algorithm designed to find approximate solutions to combinatorial optimization problems. It alternates between two unitaries $ U_C(\gamma)$ and $U_B(\beta)$ applied to a suitable initial state $| s \rangle.$ Where:

\begin{equation}
| s \rangle = | + \rangle | + \rangle \ldots | + \rangle
\end{equation}
\begin{equation}
U_B(\beta) = e^{-i\beta H_B}, \quad U_C(\gamma) = e^{-i\gamma H_C},
\end{equation}
$H_B$ is a \emph{mixing} Hamiltonian and $H_C$ encodes the cost function of the problem. Typically, QAOA relies on classical optimization to find good parameters $(\beta, \gamma)$ that optimize the target cost function. The cost Hamiltonian can be generically expressed as an Ising Hamiltonian with local fields $h_i\,\sigma_i^z$ , Ising-like interactions $J_{ij}\,\sigma_i^z \sigma_j^z$ and constant energy off-set $\Delta$ with $\sigma_i^z$ denoting Pauli Z-matrix acting on the qubit $i$.  One can then interpret the problem as a graph such that $(i,j) \in E$ denotes the edges of the graph and $i \in V$ denote the nodes, leading to the following Hamiltonian:
\begin{equation}
H_C \;=\; \sum_{i \in V} C_i \;+\; \sum_{(i,j) \in E} C_{ij}  \;+\; \Delta.
\end{equation}


Analytical formulas for the exact computation of QAOA expectation values at $p=1$ were derived in previous works~\cite{Ozaeta_2022,Wang_2018}. These can be used to determine the optimal parameters $(\beta, \gamma)$ and the corresponding optimal expectation values. A quantum computer is then required to produce similar expectation values when the QAOA circuit is run with the same parameters. Thus, probing how closely a quantum hardware is able to reproduce these analytically determined expectation values across a problem class gives \emph{a more direct probe} of hardware maturity and suitability for the problem class. 

\vspace{-0.12cm}
\subsection{Exact Stationary-Point Analysis for Single-Layer QAOA}
\label{sec:exact_p1}

For a single-layer ($p=1$) QAOA, the cost expectation
$\langle H_{C}\rangle_{p=1}(\beta,\gamma)$
and its gradient admit analytic expressions for \emph{arbitrary}
quadratic Ising Hamiltonians.  Write Equation (3) in terms of Pauli operators as:
\begin{equation}
  H_{C}\;=\;\sum_{i<j}J_{ij}\,Z_iZ_j+\sum_{i}h_i\,Z_i+ \Delta
\end{equation}
with $J_{ij},h_i\in\mathbb{R}$ ~~\cite{Ozaeta_2022,Wang_2018}.
Also let $\mathcal{G}=(V,E)$, with $|V|=N$ qubits and $|E|=M$ couplings denote the the interaction graph of the problem. Restrict
the variational angles to
\(
  (\beta,\gamma)\in[0,\pi]\times[0,2\pi];
\)

and define the following gate convention
$U_{ZZ}(\gamma)=e^{-i\gamma Z\otimes Z}$,
$U_{X}(\beta)=e^{-i\beta X}$,
one finds\cite{Ozaeta_2022}

\begin{IEEEeqnarray}{rCl}
\langle H_{C}\rangle_{p=1}(\beta,\gamma)
  &=& \sin(2\beta)\Bigl[
        \sum_{(i,j)\in E}J_{ij}\,\sin\!\bigl(2\gamma J_{ij}\bigr)
        \nonumber\\
  &&{}+\sum_{i\in V}h_i\,\cos\!\bigl(2\gamma h_i\bigr)
      \Bigr] + c,
  \label{eq:p1_energy}
\end{IEEEeqnarray}

Further define the following identities:
\begin{align}
  F_{1}(\gamma) &:= \sum_{(i,j)}J_{ij}\sin\!\bigl(2\gamma J_{ij}\bigr)
                 +\sum_{i}h_i\cos\!\bigl(2\gamma h_i\bigr),\\
  F_{2}(\gamma) &:= \sum_{(i,j)}J_{ij}^{\,2}\cos\!\bigl(2\gamma J_{ij}\bigr)
                 -\sum_{i}h_i^{\,2}\sin\!\bigl(2\gamma h_i\bigr).
\end{align}

\begin{IEEEeqnarray}{rCl}
  \frac{\partial \langle H_{C}\rangle}{\partial \beta}
    &=& 2\,\cos(2\beta)\,F_{1}(\gamma),
  \label{eq:p1_dbeta_short}\\
  \frac{\partial \langle H_{C}\rangle}{\partial \gamma}
    &=& 2\,\sin(2\beta)\,F_{2}(\gamma).
  \label{eq:p1_dgamma_short}
\end{IEEEeqnarray}

Each $F_k$ is then a trigonometric polynomial with total degree
$2M\!+\!N$; hence (by, e.g., \cite{Zygmund2002}) it has at most $2(2M\!+\!N)$ zeros in
$\gamma\in[0,2\pi)$.

\begin{theorem}[Finite stationary set, weighted couplings]
\label{thm:finite_stat_weighted}
Over $[0,\pi]\times[0,2\pi]$ the landscape
$\langle H_{C}\rangle_{p=1}$ has at most $ 5\,(2M+N)$ isolated stationary points.
\end{theorem}

\begin{proof}
\textbf{Case 1:} $\cos(2\beta)=0$.
This occurs at
$\beta\in\{\tfrac{\pi}{4},\tfrac{3\pi}{4}\}$.
Here $\partial_{\beta}=0$ automatically and
$\partial_{\gamma}=0$ reduces to $F_{2}(\gamma)=0$,
which has at most $2(2M+N)$ solutions by Ch. 4 of
\cite{Zygmund2002}.
Hence Case 1 contributes $\le 2\times 2(2M+N)=4M+2N$ points.

\textbf{Case 2:} $\sin(2\beta)=0$.
This gives $\beta\in\{0,\tfrac{\pi}{2},\pi\}$.
Then $\partial_{\gamma}=0$ holds identically and
$\partial_{\beta}=0$ demands $F_{1}(\gamma)=0$,
again with $\le 2(2M+N)$ roots. Thus Case 2 contributes
$\le 3\times 2(2M+N)=6M+3N$ points.

Adding both cases yields the stated upper bound
$10M+5N = 5(2M+N)$.
\end{proof}

A grid with $\Delta\beta=\tfrac{\pi}{4}$ and
$\Delta\gamma=\tfrac{\pi}{\,2(2M+N)}$
samples every stationary point. Even for the large instances in Table~\ref{tab:tsp_param}, the exhaustive sweep
remains modest ($\sim10^{4}$ points) and scales at most 
quadratically with~$n$, matching the $5(2M+N)$ bound of
Theorem~\ref{thm:finite_stat_weighted}. At  ost additional Order $(2M+N)$ is added from the case $ F_{1}(\gamma)=F_{2}(\gamma)=0$.

For faster hardware probes, the \textbf{s}implicial \textbf{h}omology \textbf{g}lobal \textbf{o}ptimiser (\textsc{SHGO})~\cite{Endres2018} presents a robust alternative to exhaustive grid sampling. It possesses similar analytic and deterministic bounds and yet needs \emph{far} fewer function
evaluations than the tensor-product grid once the number of bins
$n$ exceeds a problem–dependent threshold  $  n_{\min}  \propto L\!\bigl(1+\ln |E|\bigr)$.


To derive the conditions we need the following Lemma:
\begin{lemma}[Lipschitz constant]
\label{lem:lipschitz_general}
On $[0,\pi]\times[0,2\pi]$ the map
$(\beta,\gamma)\mapsto\langle H_{C}\rangle_{p=1}$
is $L$-Lipschitz with
\(
  L = 2\max\!\{\sum_{(i,j)}|J_{ij}|+\sum_{i}|h_i|,
               \sum_{(i,j)}J_{ij}^{2}+\sum_{i}h_i^{2}\}.
\)
\end{lemma}

\begin{theorem}[Sobol\,+\,SHGO evaluation bound]
\label{thm:shgo_bound_general}
Let $\varepsilon>0$ and suppose the local optimiser reaches any basin’s
optimum within $\varepsilon/2$.
If
\(
  n \ge (4L/\varepsilon)\bigl(1+\ln M\bigr)
\)
points (Called Sobol points)
are used, the SHGO algorithm\cite{Endres2018}
returns an $\varepsilon$-optimal energy after at most
$\mathcal{O}(n\ln n)$ evaluations.
\end{theorem}

\noindent\textit{For a complete proof, see}
\cite{Endres2018}[Sec.~3.3, Thms.~3–4].

 For Lipschitz-continuous objectives on a compact domain SHGO visits \emph{every}
basin once $n$ exceeds a finite problem-dependent threshold.
In practice $n\le512$ is competitive against exhaustive grid sampling.  Table~\ref{tab:grid-vs-shgo} shows that SHGO achieves competitive energies (up to $\varepsilon$) while controlling runtime. 
The explicit stationary-point set
(Theorems~\ref{thm:finite_stat_weighted} and ~\ref{thm:shgo_bound_general}) turn $p{=}1$ QAOA into a calibrated benchmark akin to Clifford fidelity tests yet rooted in an industrial workload. 

\noindent\textbf{Stationary-point detection.}
A uniform grid with spacings chosen via the Lipschitz constants of 
$\partial_\beta\langle H_C\rangle$ and $\partial_\gamma\langle H_C\rangle$ forms an $\varepsilon$-net of $[0,\pi]\times[0,2\pi]$.
Hence there exists a sampled $(\hat\beta,\hat\gamma)$ with 
$|\langle H_C\rangle(\hat\beta,\hat\gamma)-\langle H_C\rangle^*|\le\varepsilon$. With the same grid, discrete sign changes in the sampled partial derivatives bracket every stationary point to within the grid cell; a local refinement (or SHGO) then recovers the root to tolerance.

\subsection{Why Single‑Layer QAOA is Benchmark‑Grade}
\label{subsec:qaoa_benchmark_grade}

As highlighted above, most benchmarking studies probe \textit{generic} entanglement or gate error but remain disconnected from any
industrial workload.  Conversely, VQE or deep QAOA require
device‑in‑the‑loop parameter optimisation, masking hardware faults behind
classical optimiser noise. In this regard, single‑layer QAOA strikes a unique sweet spot:

\begin{itemize}
    \item \textbf{Analytic angles.}  For Ising‑type Hamiltonians the optimal
          $(\beta^{\star},\gamma^{\star})$ are known in closed form%
         \cite{Ozaeta_2022}.  The benchmark therefore spends \emph{no} quantum
          budget on angle tuning.
    \item \textbf{O(1) circuit depth.}  Depth grows linearly with
          $\operatorname{deg}(G)$ but is independent of $N$, so increasing the
          qubit count isolates \emph{coherence‑time} limitations.
    \item \textbf{Problem‑embedded connectivity.}  Every two‑qubit ZZ term is a
          real edge of the optimisation graph; “long‑range” couplings are thus
          application dictated, not synthetic.  An improved layout or routing
          algorithm rightfully increases the feasible instance size.
    \item \textbf{Scalar oracle.}  Because the ideal cost
          $\langle H_C\rangle_{p=1}$ is a single floating‑point number, the
          benchmark avoids storing $2^{N}$‑sized distributions; memory usage is
          constant.
\end{itemize}

These properties make $p=1$ QAOA a \emph{direct, scalable, and workload‑driven}
probe of hardware maturity—complementary to but more informative than gate‑only
metrics.

 \subsection{Run-to-Failure Strategy and Repeated Sampling}
\label{subsec:run_to_failure}

During a hardware probe, to ensure statistical significance, each QAOA circuit should be executed  \(R\) times (\(R=1\) to \(10\) runs in our setup), with each run collecting \(M\) measurement shots (\(5{,}000\) in our setup). This approach provides confidence intervals on the measured distribution and mitigates transient hardware fluctuations. Such repeated measurements are a common practice in reliability engineering to assess system dependability~\cite{OConnor2012,Modarres2009}. Starting from small instances, we systematically escalate the problem size \(N\) (and thus circuit size) until performance drops below the feasibility threshold \(\tau\) as depicted in Fig~\ref{fig:run_to_failure}.  Essentially, we \emph{push the device} until it reaches its randomization limit for each problem domain, a strategy reminiscent of dependability assessments in classical reliability engineering~\cite{OConnor2012,Modarres2009}.

\vspace{0.1cm}
\section{Problem Setup}

We focus on two canonical NP-hard workloads that map to
Ising Hamiltonians and therefore to single-layer QAOA.

\paragraph{MaxCut.}
Given a weighted graph \(G=(V,E,w_{ij})\), MaxCut seeks a binary
assignment \(x_i\in\{0,1\}\) that maximises
$
  \sum_{(i,j)\in E} w_{ij}(x_i\oplus x_j).
$
Under the spin re-encoding
\(s_i=1-2x_i\), the objective becomes the Ising cost
\(
  H_C=-\tfrac14\sum_{(i,j)}w_{ij}\,\sigma_i^z\sigma_j^z
\)
(up to a constant), making MaxCut a textbook QAOA benchmark
\cite{Lucas2014,Wang_2018}.

\paragraph{Travelling Salesman (TSP).}
For \(n\) cities with distances \(d_{ij}\), TSP minimises the tour length
\(
  \sum_{k=1}^{n}d_{\pi(k),\pi(k+1)},
\)
\(\pi(n+1)\!=\!\pi(1)\).
A one-hot QUBO uses binary variables \(x_{i,p}\) indicating
“city \(i\) at position \(p\)” (total \(n^2\) qubits).  Two quadratic
penalties enforce (i) each city appears once and
(ii) each position hosts exactly one city, while the cost term
\(
  \sum_{i\neq k}\sum_{p} d_{ik}\,x_{i,p}\,x_{k,(p+1)\bmod n}
\)
captures the tour length \cite{Lucas2014}.  The resulting Hamiltonian
differs only in coupling topology, letting us stress hardware
connectivity and SWAP overhead versus MaxCut. These two structured problems therefore provide complementary,
application-specific test beds for the hardware-maturity probe
presented in Section~\ref{sec:introduction}.

\vspace{0.1cm}
\begin{table}[ht]
\centering
\caption{Grid search vs.\ SHGO for $p{=}1$ QAOA on random Max-Cut (all
floating-point values rounded to two decimals).}
\label{tab:grid-vs-shgo}
\small
\resizebox{\columnwidth}{!}{%
\begin{tabular}{rrrrrrrrrrrrr}
\toprule
$n$ & $|E|$ & $2|E|+n$ & grid\_pts &
\multicolumn{1}{c}{Grid~[s]} &
$\beta_{\!g}$ & $\gamma_{\!g}$ & $E_{\!g}$ &
nfev\_shgo & \multicolumn{1}{c}{SHGO~[s]} &
$\beta_{\!s}$ & $\gamma_{\!s}$ & $E_{\!s}$\\
\midrule
16 &  63 & 142 &  2845 &  12.25 & 0.79 & 4.08 &  78.52 & 1485 &  12.04 & 1.27 & 0.13 &  90.64 \\
17 &  83 & 183 &  3665 &  23.33 & 0.79 & 5.25 &  98.81 & 1529 &  18.51 & 1.29 & 0.12 & 111.14 \\
18 & 107 & 232 &  4645 &  44.84 & 0.79 & 1.81 & 129.82 & 1652 &  30.79 & 2.89 & 0.10 & 142.75 \\
19 & 100 & 219 &  4385 &  35.57 & 0.79 & 5.90 & 127.14 & 1634 &  25.60 & 2.87 & 0.10 & 141.66 \\
20 & 114 & 248 &  4965 &  49.07 & 0.79 & 5.65 & 138.51 & 1836 &  35.37 & 2.87 & 0.11 & 154.33 \\
21 & 136 & 293 &  5865 &  76.88 & 0.79 & 5.69 & 163.76 & 1284 &  32.93 & 2.88 & 0.10 & 180.18 \\
22 & 129 & 280 &  5605 &  65.35 & 0.79 & 6.01 & 161.53 & 1531 &  34.40 & 1.29 & 0.10 & 180.08 \\
23 & 157 & 337 &  6745 & 108.44 & 0.79 & 1.72 & 197.11 & 1667 &  51.40 & 2.88 & 0.09 & 216.46 \\
24 & 158 & 340 &  6805 & 105.66 & 0.79 & 2.54 & 204.14 & 1429 &  39.48 & 1.30 & 0.09 & 224.91 \\
25 & 183 & 391 &  7825 & 140.34 & 0.79 & 1.66 & 231.78 & 1601 &  56.12 & 1.32 & 0.09 & 252.58 \\
26 & 191 & 408 &  8165 & 173.54 & 0.79 & 3.16 & 242.20 & 1631 &  66.52 & 1.31 & 0.09 & 264.87 \\
27 & 221 & 469 &  9385 & 261.16 & 0.79 & 5.10 & 267.63 & 1650 &  89.43 & 2.90 & 0.08 & 289.98 \\
28 & 211 & 450 &  9005 & 221.28 & 0.79 & 4.66 & 259.49 & 2255 & 112.00 & 1.30 & 0.09 & 285.28 \\
29 & 253 & 535 & 10705 & 358.49 & 0.79 & 2.01 & 311.76 & 1548 & 101.82 & 1.33 & 0.08 & 335.88 \\
\bottomrule
\end{tabular}}
\end{table}


\vspace{-0.2cm}
\section{Hardware Probes and Results}
\label{sec:hardware_probe}

We first publish a \textbf{catalogue} of optimised
single–layer QAOA parameters for problems in the \textsc{QOPTLib} benchmark suite  \cite{onah2025qaoa}. We release a \textbf{machine-readable catalogue} of analytically derived optima for \emph{all} instances in the
\textsc{QOPTLib} benchmark suite\,\cite{onah2025qaoa}.  Every pair
\((\beta^{\star},\gamma^{\star})\) was located by
\textbf{exhaustively evaluating} the
\(\Delta\beta \times \Delta\gamma\) grid proven in
Sec.~\ref{sec:exact_p1}, so each entry constitutes a \emph{provable}
global optimum.  Because QOPTLib already contains instances that map to
\(>\!\!400\) logical qubits, forthcoming quantum processors can simply run
the pre-tabulated circuits and gauge fidelity by comparing their measured energies to these values. Table~\ref{tab:tsp_param} lists the subset of nine Travelling Salesman instances.

\vspace{0.1cm}
\begin{table}[ht]
\centering
\caption{$p{=}1$ QAOA catalogue for TSP. All parameters measured on
\texttt{ibm\_brisbane}.}
\label{tab:tsp_param}
\small
\resizebox{\columnwidth}{!}{%
\begin{tabular}{lrrrrrrrr}
\toprule
\textbf{Inst.} & \textbf{Qubits} & \textbf{2Q gates} & \textbf{Depth} &
\textbf{Modes} & \textbf{Grid pts} &
$\boldsymbol{\beta^\star}$ & $\boldsymbol{\gamma^\star}$ &
$\boldsymbol{E_{\text{exact}}}$ \\
\midrule
wi4  & 16  &   509  &   771  &  208 &  4\,165 & 2.356 & 1.201 & $-7.01\!\times\!10^{4}$ \\
wi5  & 25  & 1\,387  & 1\,787 &  425 &  8\,505 & 0.785 & 0.503 & $-8.00\!\times\!10^{4}$ \\
wi6  & 36  & 2\,834  & 3\,060 &  756 & 15\,125 & 0.785 & 4.428 & $-1.16\!\times\!10^{5}$ \\
wi7  & 49  & 5\,141  & 4\,706 &1\,225 & 24\,505 & 2.356 & 1.417 & $-8.22\!\times\!10^{4}$ \\
dj8  & 64  & 8\,295  & 7\,221 &1\,856 & 37\,125 & 0.785 & 0.002 & $-7.39\!\times\!10^{5}$ \\
dj9  & 81  &12\,642  &10\,695 &2\,673 & 53\,470 & 0.785 & 0.001 & $-2.37\!\times\!10^{6}$ \\
\bottomrule
\end{tabular}}
\end{table}


\vspace{0.1cm}
\begin{figure}[!t]
  \centering
    \includegraphics[width=\linewidth]{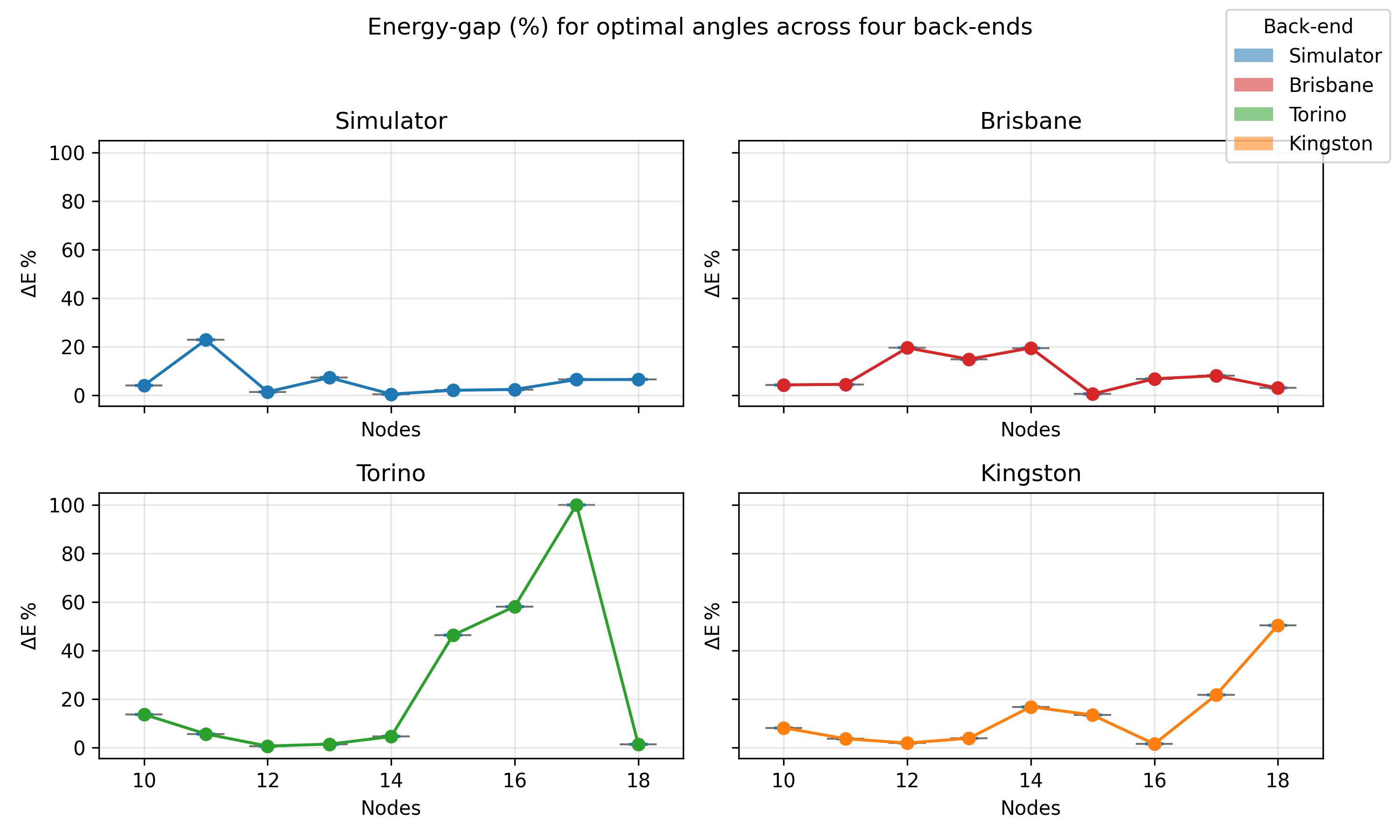}
    \caption{Energy gap $(\Delta E/E_{\text{ideal}})\times100$ using the
             analytic angles on four IBM back-ends.  The simulator stays
             within $\sim$5 \%, while hardware gaps diverge once the
             compound error $C\!\gtrsim\!1$.}
    \label{fig:energy_gap_backends}
\end{figure}

Using the parameter sets of Tables~\ref{tab:tsp_param}, we executed each circuit $M=5\,000$ times on the
\texttt{ibm\_brussels} and \texttt{ibm\_kingston} devices
as well as the \texttt{aer\_simulator}.  For TSP we apply a lightweight
\emph{feasibility repair} to every measured bit-string—flipping the
minimum number of bits to satisfy the “one city–one slot” constraints

On \texttt{ibm\_kingston} the probe recovers optimal tours for
\textsf{wi4}–\textsf{wi5} and maintains
$\Delta E/E_{\text{ideal}}\!<\!5\%$ up to 64 qubits  with occasional errors that  signify performance collapse, exactly
as predicted by the run-to-failure metric.  MaxCut problems shows a similar behaviour as documented in Figure \ref{fig:energy_gap_backends}.

\vspace{0.1cm}
\begin{table}[ht]
\centering
\small
\resizebox{\columnwidth}{!}{%
\begin{tabular}{l c c c c c c}
\toprule
Instance & $\mathcal{C}_{\text{opt}}$ & $(\gamma,\beta)$ &
QAOA Dist. & Approx. Ratio & Gap (\%) & Backend \\
\midrule
wi4  & 6700 & 2.150,2.315 & 6700 & 1.000 & 0.00  & aer\_sim      \\
     &      &             & 7149 & 0.937 & 6.70  & ibm\_brussels \\
     &      &             & 6700 & 1.000 & 0.00  & ibm\_kingston \\
\addlinespace[0.3ex]
wi5  & 6786 & 3.142,0.827 & 6888 & 0.984 & 1.51  & aer\_sim      \\
     &      &             & 8898 & 0.762 & 31.1  & ibm\_brussels \\
\addlinespace[0.3ex]
wi6  & 9815 & 3.142,0.661 &16023 & 0.612 & 63.3  & ibm\_brussels \\
     &      &             &14502 & 0.677 & 47.8  & ibm\_kingston \\
\addlinespace[0.3ex]
wi7  & 7245 & 2.150,0.827 &16753 & 0.433 &131.2  & ibm\_brussels \\
dj8  & 2762 & 2.150,2.480 & 2822 & 0.979 & 2.17  & ibm\_kingston \\
dj9  & 2134 & 2.976,0.827 & 2567 & 0.832 & 20.3  & ibm\_kingston \\
\bottomrule
\end{tabular}}
\caption{Hardware results for the TSP benchmark.
\label{tab:tsp_results}}
\end{table}

\vspace{-0.2cm}
\section{Conclusion}
\label{sec:conclusion}

We have presented a \emph{scalable, application-centred} hardware-maturity probe built on single-layer QAOA with \textbf{rigorously pre-computed parameters}.  Harmonic-analysis techniques yield closed-form bounds on the stationary points of the $p{=}1$ cost landscape, allowing the optimal angles $(\beta^\star,\gamma^\star)$ and reference energy $E_{\text{ideal}}$ to be obtained entirely offline.
The benchmark therefore measures \emph{only} what the quantum device must supply: faithful state preparation and read-out. Because the methodology uses analytic angles, touches the full quantum-classical stack, and applies unchanged to any device or workload, it constitutes a realistic, portable alternative to gate-level metrics and synthetic application suites.  We envisage its adoption as a \emph{standard dependability yard-stick} for emerging QHPC
installations, guiding error-mitigation priorities today and tracking progress toward fault-tolerant advantage tomorrow.


\vspace{-0.2cm}

\begin{thebibliography}{1}


\bibitem{preskill2018quantum}
J. Preskill, 
``Quantum Computing in the NISQ Era and Beyond,'' 
\emph{Quantum}, \textbf{2}, 79 (2018).


\bibitem{magesan2011scalable}
E.~Magesan, J.~M. Gambetta, and J.~Emerson,
\newblock ``Robust Randomized Benchmarking of Quantum Processes,''
\newblock \emph{Phys. Rev. Lett.}, vol. 106, no.~18, p. 180504, 2011.


\bibitem{EPLG21}
D.~C.~McKay, I.~Hincks, E.~J.~Pritchett, M.~Carroll, L.~C.~G.~Govia, and S.~T.~Merkel, 
“Benchmarking Quantum Processor Performance at Scale,” 
\emph{arXiv preprint arXiv:2311.05933}, 2023. 
[Online]. Available: \url{https://arxiv.org/abs/2311.05933}.


\bibitem{Clops20}
A.~Wack, H.~Paik, A.~Javadi-Abhari, P.~Jurcevic, I.~Faro, J.~M.~Gambetta, and B.~R.~Johnson,
“Quality, Speed, and Scale: three key attributes to measure the performance of near-term quantum computers,”
\emph{arXiv preprint arXiv:2110.14108}, 2021.
[Online]. \url{https://arxiv.org/abs/2110.14108}.

\bibitem{farhi2014quantum}
E.~Farhi, J.~Goldstone, and S.~Gutmann,
\newblock ``A Quantum Approximate Optimization Algorithm,''
\newblock \emph{arXiv preprint arXiv:1411.4028}, 2014.
\bibitem{Osaba_2023}
E.~Osaba and E.~Villar-Rodriguez,
\newblock ``QOPTLib: A Quantum Computing Oriented Benchmark for Combinatorial Optimization Problems,''
\newblock in \emph{Benchmarks and Hybrid Algorithms in Optimization and Applications}, Springer Nature Singapore, 2023, pp. 49–63.


\bibitem{Ozaeta_2022} 
Ozaeta, A., van Dam, W., and McMahon, P. L.,
``Expectation values from the single-layer quantum approximate optimization algorithm on Ising problems,''
\textit{Quantum Science and Technology} \textbf{7}, no. 4 (2022): 045036,
\url{http://dx.doi.org/10.1088/2058-9565/ac9013}.

\bibitem{Benedetti2020}
M. Benedetti, A. Perdomo-Ortiz, V. Leyton-Ortega, J. Realpe-Gómez, and J. I. Latorre, “Benchmarking near-term devices with quantum applications,” \textit{npj Quantum Information}, vol. 6, pp. 1–8, 2020.

\bibitem{Wang_2018} 
Wang, Z., Hadfield, S., Jiang, Z., and Rieffel, E. G.,
``Quantum approximate optimization algorithm for MaxCut: A fermionic view,''
\textit{Phys Rev A} \textbf{97}, no. 2 (2018): 022304,
\url{http://dx.doi.org/10.1103/PhysRevA.97.022304}.

\bibitem{OConnor2012}
O’Connor, P. and Kleyner, A., \textit{Practical Reliability Engineering}, 5th ed., Wiley, 2012.

\bibitem{Modarres2009}
Modarres, M., Kaminskiy, M., and Krivtsov, V., \textit{Reliability Engineering and Risk Analysis: A Practical Guide}, 3rd ed., CRC Press, 2009.
\bibitem{tomesh2022}
T.~Tomesh, P.~Gokhale, V.~Omole, G.~S.~Ravi, K.~N.~Smith, J.~Viszlai, X.-C.~Wu, N.~Hardavellas, M.~R.~Martonosi, and F.~T.~Chong,
“SupermarQ: A Scalable Quantum Benchmark Suite,”
arXiv:2202.11045 [quant-ph].

\bibitem{lubinski2023app}
T. Lubinski, S. Johri, P. Varosy, J. Coleman, L. Zhao, J. Necaise, C. H. Baldwin, K. Mayer, and T. Proctor,
``Application-Oriented Performance Benchmarks for Quantum Computing,''
\emph{arXiv:2110.03137 [quant-ph]} (2023),
\url{https://arxiv.org/abs/2110.03137}.


\bibitem{Willsch_2020} 
M. Willsch, D. Willsch, F. Jin, H. De Raedt, and K. Michielsen, 
``Benchmarking the quantum approximate optimization algorithm,'' 
\emph{Quantum Information Processing} \textbf{19}, no. 7 (2020). 
DOI: \href{http://dx.doi.org/10.1007/s11128-020-02692-8}{10.1007/s11128-020-02692-8}.
\bibitem{Michielsen_2017} 
K. Michielsen, M. Nocon, D. Willsch, F. Jin, T. Lippert, and H. De Raedt, 
``Benchmarking gate-based quantum computers,'' 
\emph{Computer Physics Communications} \textbf{220}, 44--55 (2017). 
DOI: \href{http://dx.doi.org/10.1016/j.cpc.2017.06.011}{10.1016/j.cpc.2017.06.011}.
\bibitem{MontanezBarrera2024b}
J.~A. Monta\~nez-Barrera, K.~Michielsen, and D.~E. Bernal-Neira,
\newblock ``Evaluating the performance of quantum process units at large width and depth,''
\newblock \emph{arXiv preprint}, 2024.

\bibitem{Lucas2014}
A.~Lucas,
``Ising formulations of many NP problems,''
\emph{Frontiers in Physics}, vol.~2, 2014.

%


\bibitem{Zygmund2002}
A.~Zygmund, \emph{Trigonometric Series}, 3rd~ed.\  Cambr. Univ.
  Press, 2002.
\bibitem{Endres2018}
S.~C. Endres, C.~Sandrock, and W.~W. Focke, ``A simplicial homology
  algorithm for Lipschitz optimisation,'' \emph{Journal of Global
  Optimization}, vol.~72, pp.~307–325, 2018.

\bibitem{onah2025qaoa}
C.~Onah,
“Single-Layer QAOA p = 1 Global-Optimum Catalogue (v1.0),”
Zenodo, 2025,
doi:10.5281/zenodo.15878141.


\end{thebibliography}
\end{document}